\documentclass[lettersize,journal]{IEEEtran}

\usepackage{amsmath,amssymb,amsfonts}
\usepackage{amsthm}
\usepackage{tabularx}
\usepackage[linesnumbered, ruled]{algorithm2e}
\usepackage{algorithmic}[noend]
\usepackage{array}
\usepackage[caption=false,font=normalsize,labelfont=sf,textfont=sf]{subfig}
\usepackage{textcomp}
\usepackage{stfloats}
\usepackage{url}
\usepackage{verbatim}
\usepackage{graphicx}
\usepackage{cite}
\usepackage{bm}
\usepackage{subcaption} 
\usepackage{color}

\newtheorem{theo}{Theorem}

\newtheorem{assumption}{Assumption}

\newtheorem{remark}{Remark}

\usepackage{lscape}

\usepackage{booktabs}
\usepackage{multirow}
\begin{document}
\setlength\textfloatsep{8pt}
\setlength\intextsep{8pt}
\title{Reference-Free Iterative Learning Model Predictive Control with Neural Certificates}

\author{Wataru Hashimoto, Kazumune Hashimoto, Masako Kishida, and Shigemasa Takai
        % <-this % stops a space
\thanks{{Wataru Hashimoto, Kazumune Hashimoto, and Shigemasa Takai are with the Graduate School of Engineering, The University of Osaka, Suita, Japan (e-mail: hashimoto@is.eei.eng.osaka-u.ac.jp, \{hashimoto, takai\}@eei.eng.osaka-u.ac.jp). Masako Kishida is with the National Institute of Informatics, Tokyo, Japan (email: kishida@nii.ac.jp). The corresponding author is Wataru Hashimoto.
This work is supported by JST CREST JPMJCR201, JST ACT-X JPMJAX23CK, and JSPS KAKENHI Grant 21K14184, and 22KK0155.
}}% <-this % stops a space
%\thanks{Manuscript received April 19, 2021; revised August 16, 2021.}
}



\maketitle

\begin{abstract}
In this paper, we propose a novel reference-free iterative learning model predictive control (MPC). In the proposed method, a certificate function based on the concept of Control Lyapunov Barrier Function (CLBF) is learned using data collected from past control executions and used to define the terminal set and cost in the MPC optimization problem at the current iteration.
This scheme enables the progressive refinement of the MPC's terminal components over successive iterations. Unlike existing methods that rely on mixed-integer programming and suffer from numerical difficulties, the proposed approach formulates the MPC optimization problem as a standard nonlinear program, enabling more efficient online computation.
The proposed method satisfies key MPC properties, including recursive feasibility and asymptotic stability. Additionally, we demonstrate that the performance cost is non-increasing with respect to the number of iterations, under certain assumptions. Numerical experiments including the simulation with PyBullet confirm that our control scheme iteratively enhances control performance and significantly improves online computational efficiency compared to the existing methods.

\end{abstract}


\section{Introduction}
\IEEEPARstart{M}{odel} Predictive Control (MPC) is one of the most prominent and widely studied methodologies in control literature, celebrated for its solid theoretical foundations and extensive range of applications across various domains, including industrial process control, robotics, autonomous vehicles \cite{MPC}.
%\IEEEPARstart{M}{odel} predictive control (MPC) has emerged as a powerful and versatile control strategy for a wide range of dynamical systems, owing to its ability to handle multivariate systems and incorporate constraints, and has been widely applied in industrial applications \cite{MPC}. 
However, since MPC determines a control input by solving a finite-time horizon optimal control problem at each time step, this limited foresight can lead to suboptimal decisions that do not sufficiently account for the system's long-term behavior. %As a result, maintaining stability and satisfaction of safety constraints over an extended time frame can become difficult. 

To address this issue, iterative strategies for MPC have been developed, where the control with MPC is executed iteratively for the same or similar tasks, and the control performance is incrementally improved based on the data collected from previous iterations. 
For example, in \cite{ILMPC1,ILMPC2,ILMPC3,ILMPC4,ILMPC5,ILMPC6}, the combination of iterative learning control (ILC) \cite{ILC1,ILC2} and MPC has been explored for reference tracking control problems, and it is demonstrated that the tracking error converges to zero as the number of iterations increases. In \cite{iterative1,iterative2,UgoAdd}, reference-free iterative learning MPC strategies are proposed, which refine the terminal cost and constraints of the MPC using past trajectory data. This approach theoretically and empirically ensures a non-increasing performance cost over successive iterations. By eliminating the need for a tracking reference, this method enhances the flexibility in choosing control actions compared to the reference tracking methods. 
However, since the resulting optimization involves mixed-integer programming (MIP), it poses challenges in terms of computational burden during online execution, potentially limiting its practical applicability. Moreover, the method proposed in \cite{iterative1,iterative2,UgoAdd} requires the terminal state to coincide with one of the states visited in previous iterations, which is practically challenging to implement rigorously.

To address the limitations of previous studies, this paper proposes a novel reference-free iterative learning MPC framework that utilizes neural certificate functions. In the proposed method, the neural certificate is learned from trajectory data collected in earlier MPC iterations and is used to define the terminal set and cost in MPC. This certificate is progressively refined with additional trajectory data, thereby enhancing the terminal constraint and cost in subsequent MPC iterations. Inspired by the principles of Control Lyapunov-Barrier Functions (CLBF) \cite{CLBF,NeuralCertificate2}, the neural certificate is learned to certify both the forward invariance of the safe region and the stability of the goal state.
%Our certificate function is based on Control Lyapunov Barrier Functions (CLBFs) \cite{} that integrate the concepts of Control Lyapunov Functions (CLFs) and Control Barrier Functions (CBFs)\cite{CLF,CBF}, which enable us to simultaneously certify stability and safety of a control system, respectively. The key idea of this work is the use of a CLBF learned with a Neural Network (NN) to represent the MPC terminal constraint and terminal cost. This CLBF is progressively refined with additional trajectory data, thereby enhancing the terminal constraint in subsequent MPC iterations. 
%Furthermore, the terminal cost function is iteratively learned by applying Approximate Dynamic Programming (ADP) in the region certified by the current CLBF.
With this strategy, the proposed method ensures desirable properties such as recursive feasibility, stability of the equilibrium point, and constant improvement in control performance along with the number of iterations, under certain assumptions. Moreover, since the resulting optimization problem is formulated as a standard nonlinear program, the method allows for more efficient computation during online execution compared to existing approaches, albeit at the expense of offline computation for learning the certificate function.

\textbf{Related works on iterative learning MPC:}
Control strategies for repetitive tasks that can improve control performance by effectively utilizing the data from previous experiences have been extensively studied in the literature on iterative learning control (ILC) \cite{ILC1,ILC2} for several decades.
To explicitly address the state constraints and guarantee stability of the closed-loop systems, approaches that integrate ILC with MPC have gained popularity in recent years \cite{ILMPC1,ILMPC2,ILMPC3,ILMPC4,ILMPC5,ILMPC6}. The work in \cite{ILMPC1} is one of the first to combine ILC with General Predictive Control (GPC), demonstrating significant improvements in the control performance. Further researches such as \cite{ILMPC2,ILMPC3,ILMPC4,ILMPC5,ILMPC6} consider iterative learning MPC strategies for general nonlinear systems and theoretically prove the convergence to the reference trajectory. Nonetheless, these methods rely on a predefined reference trajectory, which limits their applicability.
To achieve reference-free iterative learning MPC, the authors of the works \cite{iterative1,iterative2} proposed a way to effectively refine the terminal components of the MPC problem with trajectory data collected in past iterations. This scheme enables us to solve the infinite-time optimal control problem for linear systems without reference by iteratively performing finite-time horizon MPC, assuming that initially at least one feasible (not optimal) solution to the control problem is provided \cite{UgoOpt}. For nonlinear systems \cite{iterative1}, it ensures the performance cost is non-increasing \cite{iterative1}. This method has been further extended to uncertain linear systems \cite{iterative3}, nonlinear probabilistic systems \cite{iterative4}, unknown nonlinear systems \cite{self2}, and multi-agent systems \cite{iterative5}. It has also been tested on several challenging applications, including autonomous racing \cite{AV} and surgical robot \cite{surgical}. %As an another line of research, Reinforcement Learning (RL) based approaches for iteratively improving the MPC performance are also considered in \cite{RL1,RL2,RL3,RL4,RL5,RL6}. 
%Additionally, the approach guarantees desired MPC properties such as recursive feasibility and closed-loop system stability through the learned terminal components.

\textbf{Related works on the learning-based construction of terminal components:}
The terminal costs and constraints play a crucial role in rigorously guaranteeing the recursive feasibility and stability of MPC. Traditionally, these terminal components are designed using the Control Lyapunov Function (CLF) \cite{MPCCLF1,MPCCLF2}. 
Previous works often construct CLFs based on linearization around the equilibrium point \cite{MPCCLF1}, which restricts the resulting terminal set to a small neighborhood of the equilibrium and leads to conservative control performance and difficulty in dealing with short prediction horizons.
Consequently, learning-based approaches to building terminal components have gained increasing attention in recent years. The authors of the works \cite{RL1,RL2,RL3,RL4,RL5,RL6,RL7,RL8} consider using Approximate Dynamic Programming (ADP) or Reinforcement Learning (RL) to learn the terminal components. While these approaches have demonstrated effectiveness, they often require the full implementation of ADP or RL, which can be computationally intensive. Moreover, these methods typically do not address the explicit construction of the terminal set, leaving the safety guarantees beyond the prediction horizon either unverified or dependent on additional assumptions. The aforementioned iterative learning MPC scheme \cite{iterative1,iterative2} considers defining terminal components with state trajectories in the past iterations, which enables simpler construction of the terminal components.

\textbf{Related works on learning-based certificate functions:} 
Certificate functions, such as CLF and Control Barrier Function (CBF), are instrumental in expressing desirable properties of dynamical systems, including stability and safety \cite{CBF}. However, identifying suitable certificate functions for a system remains a complex and challenging task in general. %Recent advances have proposed methods for the automatic synthesis of certificate functions \cite{review}. %For systems with polynomial dynamics, this synthesis can be formulated as a convex semi-definite optimization problem using sum-of-squares (SoS) techniques \cite{SOS}. However, SoS methods are restricted to polynomial systems and encounter scalability issues in higher-dimensional settings \cite{SOS2}. 
To overcome this problem, there has been a growing interest in employing neural networks to learn certificate functions \cite{NeuralCertificate1,NeuralCertificate2,NeuralCertificate3,NeuralCertificate4,NeuralCertificate5,NeuralCertificate6,NeuralCertificate7,NeuralCertificate8,RobustCBF,MetaCBF}. While neural network-based approaches have demonstrated flexibility in constructing certificates and achieving larger region of attraction compared to other methods, several practical challenges persist. One notable issue is the difficulty in collecting training samples. while most of the previous works assume that the samples from the safe region are freely available \cite{NeuralCertificate1,NeuralCertificate2,NeuralCertificate3,NeuralCertificate4,NeuralCertificate5} or expert trajectories are given \cite{NeuralCertificate6,RobustCBF}, finding safe regions for sampling or obtaining expert trajectories is often not straightforward. 
%even when the controller has complete knowledge of unsafe regions such as obstacles, due to the effect of the dynamics.
%Specifically, the safe region is often unknown a priori, even when the controller has complete knowledge of geometric unsafe regions such as obstacles because the system may still enter an unsafe state in the future depending on the dynamics, even if the current state is outside these obstacles.

\textbf{Contributions:} The contributions of our proposed method, along with a comparison to existing approaches, are summarized in the following. %First, this paper introduces a reference-free iterative MPC strategy for initially unknown environments that employs a neural certificate function learned from data collected in previous iterations, as both the terminal constraint and terminal cost. This approach guarantees desirable properties for MPC, including recursive feasibility, stability, and consistent improvement in control performance along with the number of iterations. While the method proposed in \cite{iterative1} also iteratively improves the terminal set and cost, it leads to a mixed-integer programming problem, which is computationally inefficient. In contrast, our method results in a standard nonlinear programming problem, significantly reducing computation time in online execution. While the authors of \cite{RL5} employs adaptive dynamics programming to learn terminal cost and achieves improving the efficiency of the online computation, the learned function does not take into account safety constraint.
%This paper introduces a reference-free iterative MPC strategy designed for initially unknown environments. 
First, this paper proposes a reference-free iterative learning MPC strategy that utilizes a neural certificate function, learned from data collected in previous iterations, as both the terminal constraint and terminal cost. This strategy facilitates the iterative enhancement of control performance as the number of iterations increases, while guaranteeing essential MPC properties such as recursive feasibility and closed-loop stability. Unlike the method proposed in \cite{iterative1}, which iteratively improves the terminal set and cost but results in a computationally expensive mixed-integer programming problem, our method simplifies the problem to a standard nonlinear programming problem. This reduces computation time during online control execution. Moreover, the proposed method does not require full implementation of RL or dynamic programming as the previous works \cite{RL1,RL2,RL3,RL4,RL5,RL6,RL7,RL8}. %Moreover, while the approach in \cite{RL5} employs adaptive dynamic programming to learn terminal costs and improves online computational efficiency, it neglects safety constraints, limiting its applicability in safety-critical tasks.
%Second, our approach is also advantageous from the perspective of neural certificate function literature in the sense that the proposed iterative MPC formulation enables the system to explore unseen safe areas and provides an effective data collection procedure for learning certificates. 

Second, the proposed method is also potentially advantageous from the perspective of neural certificate function literature. Specifically, the proposed iterative learning MPC formulation facilitates the exploration of previously unseen safe regions and thus enables a system to have a systematic data collection process for learning certificates without the need for expert trajectories or explicit knowledge of safe invariant regions.
%Many prior studies on neural certificate functions \cite{review} assume that safe/unsafe regions are initially known and samples from these sets are readily available. However, in many applications, the safe region corresponding to system dynamics and safety requirements is not known in advance. 
%we can solve the numerical issue in \ref{} associated with the  

Lastly, the simulation study, including the experiment with PyBullet simulator \cite{pybullet} demonstrates that the proposed method iteratively enhances control performance while significantly improving online computational efficiency compared to existing approaches.
\begin{comment}
\begin{itemize}
    \item The proposed method considers how to improve terminal components of the MPC problem for unknown nonlinear systems with the learned uncertainty propagation models based on \cite{} to achieve desired properties .
    \item We introduce the technique called Backward Enlargement to improve the applicability of the method.
    \item We test the proposed method with the simulation study on .
\end{itemize}
\end{comment}
\section{Preliminaries}\label{pre}
In this section, we summarize some preliminaries including the system descriptions and the goal of this paper.
\subsection{System Description}
We focus on the control of a nonlinear, discrete-time dynamical system, which is expressed in the general form:
\begin{align}\label{dynamics}
    &x_{t+1} = f(x_t, u_t), \quad x_t \in \mathcal{X}, \ u_t\in \mathcal{U},
\end{align}
where \(x_t \in \mathcal{X} \subseteq \mathbb{R}^n\) and \(u_t \in \mathcal{U} \subseteq \mathbb{R}^m\) denote the state vector and control input at a discrete time step \(t \in \mathbb{N}\). The sets \(\mathcal{X}\) and \(\mathcal{U}\) define the domain of interest in state space and control input constraints, respectively. 
The function \(f: \mathbb{R}^n \times \mathbb{R}^m \rightarrow \mathbb{R}^n\) defines the system dynamics, mapping the current state and control input to the next state. We assume that $f$ is continuous. In addition, we denote a set of unsafe sets to avoid (e.g., obstacle regions) by $\mathcal{A}\subset \mathcal{X}$. 
\subsection{Goal of this paper}
We first consider the following infinite-time optimal control problem:
\begin{subequations}\label{problem}
\begin{align}
J^*_{0\rightarrow \infty}(x_s) =& \min_{u_0,u_1,\ldots}\sum_{k=0}^{\infty} \gamma^k \ell(x_k,u_k),\label{prob1:cost}\\
&\mathrm{s.t.}\  x_{k+1} = {f}(x_k,u_k),\quad \forall k\geq 0, \label{prob1:system}\\
&\qquad x_0=x_s\in \mathcal{X}\backslash \mathcal{A},\\
&\qquad x_k  \in \mathcal{X}\backslash \mathcal{A}, \quad \forall k\in \{1,2,\ldots \}, \label{prob1:state const}\\
& \qquad u_k \in \mathcal{U}, \quad \forall k\in \{0,1,\ldots \}, \label{prob1:input const}
\end{align}
\end{subequations}
where $x_s\in\mathcal{X}\backslash \mathcal{A}$ is the initial state, the function $\ell:\  \mathbb{R}^{n}\times \mathbb{R}^{m}\rightarrow \mathbb{R}_{\geq 0}$ represents the stage cost of the control problem, and $\gamma$ is the discount factor with $0<\gamma<1$. 
We impose the following assumptions on the function $\ell$ and the optimization problem (\ref{problem}), which are fairly common in the optimal control literature.
\begin{assumption}\label{assumption:stagecost}
    The stage cost function $\ell$ is continuous and satisfies the following.
    \begin{align}
        &\ell(x_F,0) = 0, \\
        &\ell(x_t,u_t)> 0,\ \forall x_t\in \mathbb{R}^n\backslash \{x_F \},\ u_t\in \mathbb{R}^m\backslash \{0\},
    \end{align}
    where the final state $x_F\in \mathcal{X} \setminus \mathcal{A}$ is assumed to be an equilibrium of the unforced system (\ref{dynamics}), i.e., $f(x_F,0)=x_F$.
\end{assumption}
\begin{assumption}\label{assumption:feasibility}
    A local optimal solution to (\ref{problem}) exists.
\end{assumption}
Since directly solving the optimization problem (\ref{problem}) is challenging, we approach it by iteratively executing finite-horizon MPC. At each iteration $j\geq1$, we perform control with MPC and collect trajectory data, \(\{x_t^j, u_t^j\}_{t=0}^\infty \), where \(x_t^j\) and \(u_t^j\) represent the state and control input at time \(t\) in iteration \(j\). Then, the MPC formulation is improved with the collected data.
Throughout this paper, we assume that the initial state is fixed across the iterations:
\begin{assumption}\label{assumption:initial}
We assume that the initial state at each iteration is fixed to \(x_s \in \mathcal{X}\), i.e., $x_0^j = x_s \in \mathcal{X} \setminus \mathcal{A}, \ \forall j \in \mathbb{N}_{\geq 1}$.
\end{assumption}

Such a setting is previously considered in the iterative learning MPC framework proposed in \cite{iterative1,iterative2}. In these studies, the authors use the fact that the set consisting of all the samples collected in the previous iterations $\mathcal{SS}^j=\{ \{x_t^i\}_{t=0}^\infty \}_{i=0}^j$ is a subset of the safe maximal stabilizable set to $x_F$ (i.e., the set with maximum volume from which there exists a control sequence that can drive the system to $x_F$ while ensuring constraints (\ref{prob1:state const}) and (\ref{prob1:input const})), and use it as the terminal set of the MPC. However, these approaches often require the terminal state to exactly match one of the stored states in $\mathcal{SS}^j$, resulting in a computationally demanding mixed-integer programming problem.

Motivated by the above discussion, the objective of this paper is to develop a novel iterative learning MPC framework with the following properties: (i) The closed-loop system converges asymptotically to the target state $x_F$ (ii) The constraints (\ref{prob1:state const}) and (\ref{prob1:input const}) are satisfied at all the time instances (iii) The performance cost denoted as $J_{0\rightarrow \infty}^j(x_s) = \sum_{t=0}^{\infty}\gamma^t\ell(x_t^j,u_t^j)$  is non-increasing with respect to the iteration number $j$ (iii) the resulting optimization is formulated as a standard nonlinear programming problem, which can be efficiently solved.
In the following, Section~\ref{proposed} introduces the proposed iterative learning MPC scheme to this end, followed by a discussion of its theoretical properties in Section \ref{sec:theo}.
%\begin{assumption}\label{assumption:initial}
%We assume that the initial state at each iteration is fixed to \(x_s \in \mathcal{X}\), i.e., $x_0^j = x_s \in \mathcal{X} \setminus \mathcal{A}, \ \forall j \in \mathbb{N}_{\geq 1}$.
%\end{assumption}}
\begin{remark}
    In practice, each iteration has a finite-time duration. However, for the sake of analytical simplicity, the literature frequently employs an infinite time formulation for each iteration. In this paper, we adopt the same approach, and this choice does not affect the practical applicability of our proposed method.
\end{remark}

\begin{remark}
   For simplicity in the theoretical analysis, we assume a common initial state across all iterations as described in Assumption \ref{assumption:initial}. However, the proposed method remains applicable even with varying initial states, as long as the optimization problem is feasible at the initial time. In this case, recursive feasibility and the stability of $x_F$ of the proposed MPC scheme are still ensured (see Section \ref{sec:theo}).
\end{remark}

\section{Proposed Method}\label{proposed}
\begin{figure}[tb]
 \begin{center}
  \includegraphics[width=1\hsize]{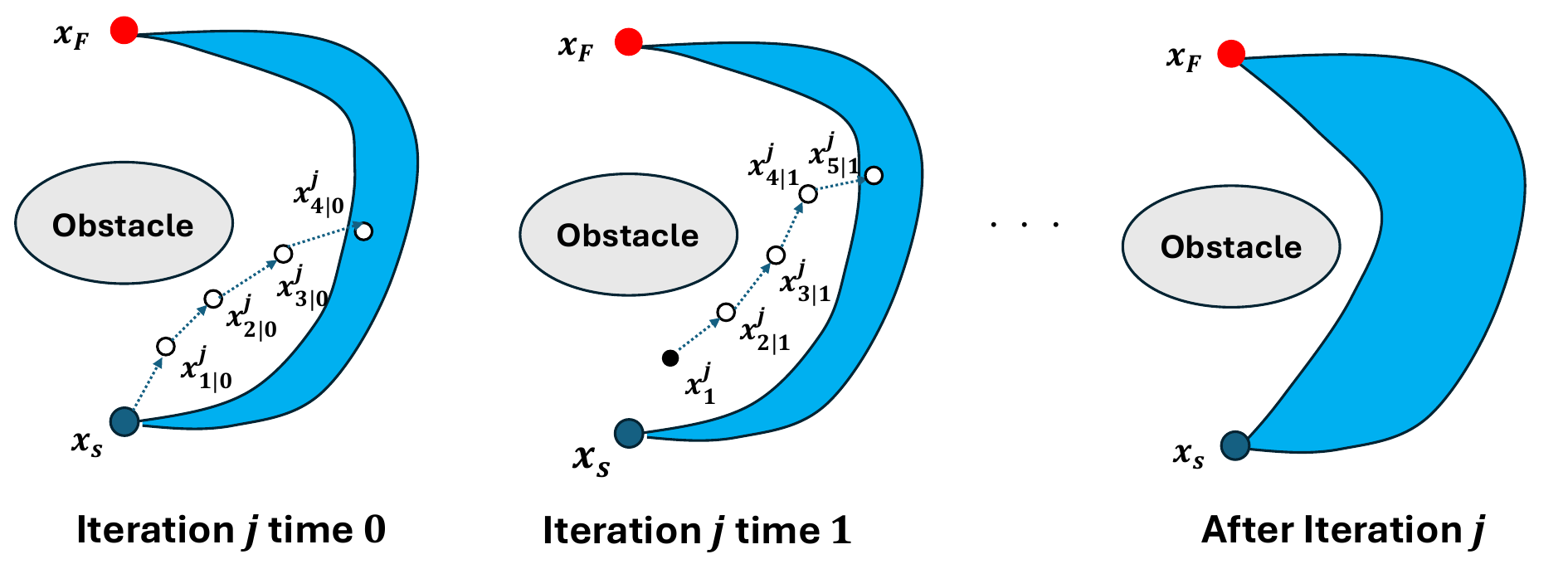}
 \end{center}
 \caption{The illustrative explanation of the proposed MPC scheme. The blue-colored regions represent the terminal set, i.e., the set $\{x\in \mathcal{X}\mid V_{\theta_{j-1}}\leq c\}$.
 At time $t$ in iteration $j$, the optimization (\ref{problem2}) is solved, and the control input (\ref{MPC input}) is applied to the system (\ref{dynamics}) and the next state $x_{t+1}^j$ is observed. Then (\ref{problem2}) is solved with the initial state $x_{t+1}^j$ (left and middle). 
 This process is repeated until convergence. After iteration $j$, the terminal function is updated based on the data collected in the past iterations (right).}
 \label{fig:terminal}
\end{figure}
In this section, we explain the proposed iterative learning MPC scheme to achieve the goal discussed in Section \ref{pre}. First, we introduce the formulation of the proposed MPC optimization problem and explain the overall iterative learning MPC procedure in Section \ref{subsec:formulation}. Then, the detailed construction method of the terminal function is discussed in Section \ref{subsec:terminalcost}.
\subsection{MPC Formulation}\label{subsec:formulation}
The proposed MPC optimization problem is defined as follows:
\begin{subequations}\label{problem2}
\begin{align}
&J^{\mathrm{LMPC},j}_{t\rightarrow t+N}({x}_t^j) = \min_{{u}^j_{t|t},\ldots,{u}^j_{t+N-1|t}}\left[ \sum_{k=t}^{t+N-1} \gamma^k \ell({x}^j_{k|t},{u}^j_{k|t})\right.\notag \\
& \qquad \qquad \qquad \left. + \gamma^{t+N} V^{j-1}(x_{t+N|t}^j) \right]\\
&\mathrm{s.t.}\notag \\
&\quad {x}^j_{k+1|t} = {f}({x}^j_{k|t},{u}^j_{k|t}),\  \forall k\in [t,\ldots,t+N-1],\label{const dyn}\\
&\quad {x}^j_{k|t} \in \mathcal{X}\backslash \mathcal{A},\ \forall k\in [t,\ldots,t+N-1],\label{state const}\\
&\quad {u}^j_{k|t} \in \mathcal{U},\ \forall k\in [t,\ldots,t+N-1],\label{input const}\\
&\quad V^{j-1}(x_{t+N|t}^j)\leq c \label{terminal},\\
&\quad {x}^j_{t|t} = {x}_t^j\label{const init},
\end{align}
\end{subequations}
where (\ref{const dyn})  and (\ref{const init}) represent the system dynamics and initial condition, respectively. The state and input constraints are given by (\ref{state const}) and (\ref{input const}) respectively. The constraint (\ref{terminal}) is the terminal constraint that enforces the terminal state into the safe invariant set defined by $V^j$. The terminal cost is also represented by the function $V^j$. The concrete definition and construction method of the function $V^j$ is discussed in Section \ref{subsec:terminalcost}.
We denote the optimal solution and corresponding predicted state trajectory at time $t$ in iteration $j$ as
\begin{align}
    {u}^{j,*}_{t:t+N-1|t} &= [{u}^{j,*}_{t|t},\ldots, {u}^{j,*}_{t+N-1|t}].\label{sol inputs}\\
     {x}^{j,*}_{t+1:t+N|t} &= [{x}^{j,*}_{t+1|t},\ldots, {x}^{j,*}_{t+N|t}].\label{sol states}
\end{align}
Then, at each time $t$ in iteration $j$, the following control input is applied to the system (\ref{dynamics}):
\begin{align}\label{MPC input}
u_t^j =u^{j,*}_{t|t}.
\end{align}
After applying the control input $u_t^j$ to the system (\ref{dynamics}), the state at the next time step $x_{t+1}^j$ is observed. Then, the finite time optimal control problem (\ref{problem2}) is solved again at time $t+1$, from the new initial state $x^j_{t+1|t+1}=x_{t+1}^j$, yielding a receding horizon control strategy. After the iteration, the dataset for learning the function $V^j$ is updated as $\mathcal{D}_{j}\leftarrow \mathcal{D}_{j-1}\cup \{x_t^j,u_t^j\}_{t=0}^\infty$. Then, $V^j$ is updated based on the data $\mathcal{D}_j$ with the procedures discussed in Section \ref{subsec:terminalcost}. Since the proposed iterative learning MPC scheme cannot be executed without an appropriate initial terminal function, $V^0$, we impose the following assumption.

\begin{assumption}\label{initialCLBF}
    The initial dataset $\mathcal{D}_0$ for learning the initial terminal function $V^0$ is given. The certified region defined by $V^0$, $\{x \in \mathcal{X} \mid V^0(x) \leq c\}$, is non-empty and includes all the points in $\mathcal{D}_0$. Additionally, the problem (\ref{problem2}) is feasible at time 0.
\end{assumption}
The whole proposed control scheme is summarized in Algorithm \ref{alg}.
In the following subsection, the construction method of the function $V^j$ used to define the terminal set and cost is explained. %Then, the theoretical properties for the proposed control scheme is discussed in Section \ref{sec:theo}.
\begin{algorithm}[t]
{\small
\SetKwInOut{Input}{Input}
\SetKwInOut{Output}{Output}
\Input{ $x_s$ (initial state); $f$ (dynamics); $V^0=V_{\theta_0}$ (initial terminal function); $N_{\mathrm{ite}}$ (number of iterations); $N$ (horizon length); $T$ (execution time steps); $\mathcal{D}_0$ (initial dataset); $\gamma$ (discount factor)}

\For{$j=1,2,\ldots,N_{\mathrm{ite}}$}{
Set the initial state by $x_0^j=x_s$;\\
\For{$t=0,1,\ldots,T$}{

Solve (\ref{problem2}) and obtain the solutions (\ref{sol inputs}) and (\ref{sol states});\\
Apply the control input (\ref{MPC input}) to the system (\ref{dynamics}) and observe the next state $x_{t+1}^j$;\\

}
Update the dataset: $\mathcal{D}_{j}\leftarrow \mathcal{D}_{j-1}\cup \{x_t^j,u_t^j\}_{t=0}^\infty$;\\
Update $V^{j-1}=V_{\theta_{j-1}}$ to $V^j=V_{\theta_{j}}$ based on the data $\mathcal{D}_j$ with the procedure in Section \ref{subsec:terminalcost};

}

    \caption{The proposed control strategy}\label{alg}
    }
\end{algorithm}

\begin{algorithm}[t]
\caption{Learning terminal function $V^j$}
\label{alg:training_CLBF_compact}
\KwIn{$\mathcal{D}_j$ (trajectory data at iteration $j$); $f$ (dynamics); $x_F$ (terminal state); $a_1-a_5$, $c$ (tuning parameters); $\gamma$ (discount factor); $k_{\mathrm{val}}$ (time interval for validation)}

Initialize the NN parameters $\theta_j$ and $\phi_j$\;
Construct the $\alpha$-shape boundary $\mathcal{B}_{\alpha}$ using $\mathcal{D}_j$\;
Define $\mathcal{X}_{\mathrm{safe}} = \{x \mid x \notin \mathcal{B}_{\alpha} \}$, $\mathcal{X}_{\mathrm{unsafe}} = \{x \mid x \in \mathcal{B}_{\alpha} \}$\;
Construct sets of samples $\mathcal{D}_{\mathrm{safe}} \subset \mathcal{X}_{\mathrm{safe}}$ and $\mathcal{D}_{\mathrm{unsafe}} \subset \mathcal{X}_{\mathrm{unsafe}}$\;

\For{$k = 1$ \textbf{\textrm{to}} $N_{\mathrm{iter}}$}{
    Compute loss (\ref{lossCLBF}) with $\mathcal{D}_{\mathrm{safe}}$ and $\mathcal{D}_{\mathrm{unsafe}}$\; 
    Update the parameters $\theta_j$ and $\phi_j$ with SGD or adam\;
    \If{$k \bmod k_{\mathrm{val}} = 0$}{
        Construct sample set $\mathcal{D}_{\mathrm{val}}\subset \mathcal{X}$ and identify $\mathcal{C}_{\mathrm{val}} \subseteq \mathcal{D}_{\mathrm{val}}$ violating (\ref{cond:CLBF1})-(\ref{cond:CLBF6})\;
        Update sample sets: 
        $\mathcal{D}_{\mathrm{safe}} \gets \mathcal{D}_{\mathrm{safe}} \cup (\mathcal{C}_{\mathrm{val}} \cap \mathcal{X}_{\mathrm{safe}})$, 
        $\mathcal{D}_{\mathrm{unsafe}} \gets \mathcal{D}_{\mathrm{unsafe}} \cup (\mathcal{C}_{\mathrm{val}} \cap \mathcal{X}_{\mathrm{unsafe}})$, 
        
    }
}
\end{algorithm}

\subsection{Construction of terminal set and cost}\label{subsec:terminalcost}
In this subsection, we elaborate on how to construct the function $V^j$ in the optimization problem (\ref{problem2}) based on the trajectory data collected in the previous iterations $\mathcal{D}_j=\{\{(x_t^i,u_t^i)\}_{t=0}^\infty \}_{i=1}^{j}$ so that the terminal region $\{x\in \mathcal{X}\mid V^{j-1}(x)\leq c\}$ to be a subset of safe maximal stabilizable set to $x_F$ and yields desirable MPC properties discussed in Section \ref{pre}. We construct the function $V^j$ based on the concept of CLBF \cite{NeuralCertificate2}, which simultaneously encodes stability and safety properties. This approach enables us to avoid the computationally intensive mixed-integer programming formulation as in \cite{iterative1} and reduce the optimization problem to a standard nonlinear programming problem.
More specifically, we consider constructing a function $V^j$ that satisfies the following conditions for an appropriately chosen constant $c>0$:
\begin{subequations}
    \begin{align}
        &V^j(x_F)=0,\label{cond:CLBF1}\\
        &V^j(x)>0,\ \forall x\in \mathcal{X}\backslash x_{F},\label{cond:CLBF2}\\
        &V^j(x)\leq c,\ \forall x\in \mathcal{X}_{\mathrm{safe}},\label{cond:CLBF3}\\
        &V^j(x)>c,\ \forall x\in \mathcal{X}_{\mathrm{unsafe}},\label{cond:CLBF4}\\
        &\inf_{u\in \mathcal{U}} V^j(f(x,u))-V^j(x)\leq 0,\ \forall x\in \mathcal{X}_{\mathrm{safe}},\label{cond:CLBF5}\\
        &\inf_{u\in \mathcal{U}} \gamma V^j(f(x,u))-V^j(x)+\ell(x,u) \leq 0,\ \forall x\in \mathcal{X}_{\mathrm{safe}},\label{cond:CLBF6}\\
        &\gamma V^j(x^{j}_{k+1})-V^j(x^{j}_k)+\ell(x^{j}_k,u^{j}_k) \geq 0,\ \forall k\in \mathbb{N}_{\geq 0},\label{cond:CLBF7}
    \end{align}
    \end{subequations}
where $\mathcal{X}_{\mathrm{unsafe}}$ is some super set of $\mathcal{A}$ and $\mathcal{X}_{\mathrm{safe}}$ is a subset of the safe maximal stabilizable set to $x_F$ from which the training samples are drawn.The conditions (\ref{cond:CLBF1})-(\ref{cond:CLBF5}) are closely related to the CLBF condition in \cite{NeuralCertificate2}. We extend it from a continuous-time system to a discrete-time system. 
The conditions (\ref{cond:CLBF6}) and (\ref{cond:CLBF7}) are introduced to ensure the stability of the resulting MPC and to guarantee a non-increasing control performance with respect to the number of iterations (see Section \ref{sec:theo}).

To find the function $V^j$ that meets the conditions (\ref{cond:CLBF1})-(\ref{cond:CLBF7}) from data, we represent the function \( V^j \) and the corresponding control policy using neural networks, denoted as \( V_{\theta_j} \) and \( \pi_{\phi_j} \), with parameters \( \theta_j \) and \( \phi_j \), respectively. Then, we learn them to satisfy conditions (\ref{cond:CLBF1})–(\ref{cond:CLBF7}).
To ensure the condition (\ref{cond:CLBF2}) by construction, we define the NN structure of $V_{\theta_j}$ as $V_{\theta_j}(x)=w_{\theta_j}(x)^\top w_{\theta_j}(x)\geq 0$, where $w_{\theta_j}(x)$ is the feedforward neural network.
These parameters are then learned to minimize the following loss function:
\begin{align}\label{lossCLBF}
    &\mathrm{loss} = V_{\theta_j}(x_F)^2+\frac{a_1}{N_{\mathrm{safe}}}\sum_{x\in \mathcal{X}_{\mathrm{safe}}}[V_{\theta_j}(x) -c]_+\notag \\
    & \quad + \frac{a_2}{N_{\mathrm{unsafe}}}\sum_{x\in \mathcal{X}_{\mathrm{unsafe}}}[c -V_{\theta_j}(x)]_+\notag \\
    +& \frac{a_3}{N_{\mathrm{safe}}} \sum_{x\in \mathcal{X}_\mathrm{safe}}[ V_{\theta_j}(f(x,\pi_{\phi_j}(x)))-V_{\theta_j}(x)]_+\notag \\
    +& \frac{a_4}{N_{\mathrm{safe}}} \sum_{x\in \mathcal{X}_\mathrm{safe}}[ \gamma V_{\theta_j}(f(x,\pi_{\phi_j}(x)))-V_{\theta_j}(x)+\ell(x,\pi_{\phi_j}(x))]_+\notag \\
    &+a_5 [-\gamma V_{\theta_j}(f(x^j_k,u^j_k)+V_{\theta_j}(x^j_k)-\ell(x^j_k,u^j_k)]_+.
\end{align}
where $a_1$, $a_2$, $a_3$, $a_4$, and $a_5$ are positive tuning parameters, $N_{\mathrm{safe}}$ and $N_{\mathrm{unsafe}}$ are the number of samples from $\mathcal{X}_{\mathrm{safe}}$ and $\mathcal{X}_{\mathrm{unsafe}}$, respectively, and $[o]_+=\max (o,0)$.
The terms in this loss function are directly linked to conditions (\ref{cond:CLBF2})-(\ref{cond:CLBF7}).
This loss is optimized using a stochastic gradient-based method, such as stochastic gradient descent (SGD) or Adam. 
%\textcolor{red}{In our setting, we use state trajectories obtained from past iterations as safe samples and take unsafe samples from the unsafe set $\mathcal{A}$. However, since the samples in $\mathcal{D}_j$ tend to be sparse, utilizing some interpolation techniques is desirable in practice. We discuss such practical considerations in Section \ref{subsec:practical}.}
In our setting, we can use state trajectories obtained from past iterations, $\mathcal{D}_j$, as safe samples because of the constraints (\ref{state const}) and (\ref{terminal}), and take unsafe samples from the unsafe set $\mathcal{A}$. 
Although the samples in $\mathcal{D}_j$ may be sparse in some regions, this issue can be addressed through several interpolation techniques. We discuss such practical strategies in Section~\ref{subsec:practical}.

Notably, the proposed MPC formulation (\ref{problem2}) enforces only the terminal state to lie within the terminal set, while allowing intermediate states to leave the certified region and the execution of the MPC naturally facilitates exploration of previously unseen safe areas and enables us to collect data to improve the certificates (see also Fig \ref{fig:terminal}). Since finding the set for sampling $\mathcal{X}_{\mathrm{safe}}$ is not straightforward and many works regarding learning certificate function literature \cite{review} assume that $\mathcal{X}_{\mathrm{safe}}$ or expert trajectories are initially given, this can be seen as one of an advantage of the proposed method.

\subsubsection{Verification of $V^j$}
Since the function $V^j$ is learned by minimizing the loss function over a finite samples, its validity across the entire state space cannot be guaranteed. To address this, we employ a verification and synthesis scheme that periodically checks for the satisfaction of conditions (\ref{cond:CLBF2})--(\ref{cond:CLBF7}) and identifies samples that violate them. In this paper, at regular intervals, we sample $N_{\mathrm{test}}$ points from the state space to evaluate the satisfaction of conditions. Any samples found to violate the conditions are then added to the dataset for further training. Other verification techniques used in learning certificate function literature \cite{review} can also be used. The whole training procedures are summarized in Algorithm \ref{alg:training_CLBF_compact}. 

\subsubsection{Practical considerations}\label{subsec:practical}
%Here, we discuss how to construct the safe and unsafe samples $\mathcal{X}_{\mathrm{safe}}$ and $\mathcal{X}_{\mathrm{unsafe}}$ from the trajectory data $\mathcal{D}_j$. 
%Here, we discuss practical considerations for constructing $\mathcal{X}_{\mathrm{safe}}$ and $\mathcal{X}_{\mathrm{unsafe}}$, and learning procedure of $V^j$.
Here, we discuss some practical considerations for learning the function $V^j$.

\textbf{Construction of $\mathcal{X}_{\mathrm{safe}}$ and $\mathcal{X}_{\mathrm{unsafe}}$:}
As previously mentioned in this section, we can define the safe samples for evaluating the loss (\ref{lossCLBF}) by all of the states within $\mathcal{D}_j$ since the states in $\mathcal{D}_j$ are guaranteed to be sampled from an actual invariant safe set, and there always exists a sequence of control inputs that steer the system from these states to the goal state $x_F$ due to the constraints in MPC formulation and the properties imposed on the terminal set.  However, the sparsity of $\mathcal{D}_j$ may be problematic. 
To address this issue, we can adapt the methods that detect the geometric boundary of the safe samples in $\mathcal{D}_j$ similar to the previous work for learning certificate functions from expert trajectory \cite{RobustCBF} (see Section 6.2 in \cite{RobustCBF}). Specifically, this study employs the alpha shape \cite{AlphaShape} which is a generalization of the convex hull of a set of points, designed to capture the ``shape" of the point set in a way that can handle concavities. 
With this method, we can densely sample safe states from the interior of the alpha shape boundary and unsafe states from the exterior. 

\textbf{On loss function:} Depending on the shapes of $\mathcal{X}_{\mathrm{safe}}$ and $\mathcal{X}_{\mathrm{unsafe}}$, learning the function $V^j$ that simultaneously satisfies conditions (\ref{cond:CLBF1})–(\ref{cond:CLBF7}) can be challenging. In such cases, training performance can be improved by generating trajectories within the alpha-shape boundary that steer the system states from initial conditions in $\mathcal{X}_{\mathrm{safe}}$ to the goal state $x_F$. Training can then be guided by incorporating a loss term that penalizes deviations from the cost-to-go associated with these trajectories. Such trajectories can be obtained either by solving an optimal control problem with constraints that ensure the system remains within $\mathcal{X}_{\mathrm{safe}}$ at all times or by leveraging expert demonstrations from humans.

\begin{remark}
Although learning function $V^j$ may be time-consuming, it is important to note that this learning process can be performed offline, thereby not affecting the computational cost during online execution.
\end{remark}

\section{Theoretical Properties}\label{sec:theo}
In this section, we analyze the theoretical properties of the proposed control scheme. 
We first make the following assumptions.
\begin{assumption}
    The satisfaction of the conditions (\ref{cond:CLBF1})--(\ref{cond:CLBF4}) are verified for the learned function $V^j$ for all $j\geq1$.
\end{assumption}
\begin{assumption}\label{assumption:monotonic}
    The region certified by $V^j$ as safe is monotonically enlarged along with the iteration, i.e., $\{x\in \mathcal{X}\mid V^{j-1}(x)\leq c\}\subset \{x\in \mathcal{X}\mid V^{j}(x)\leq c\}$, $\forall j\in \mathbb{N}$. 
\end{assumption}
\begin{assumption}\label{assumption:clbf1}
    All of the states in $\mathcal{D}_{j}$ is included in the certified region $\{x\in \mathcal{X}\mid V^{j}(x)\leq c\}$ for all $j\in \mathbb{N}$.
\end{assumption}
\begin{assumption}\label{assumption:adp}
   The violations of the conditions (\ref{cond:CLBF5}) and (\ref{cond:CLBF6}) for the learned function $V^j$ and $\pi_{\phi_j}$ are bounded by the function $\delta_1$ and constant $\delta_2$ for all iterations $j\geq 1$ as follows:
   \begin{align}
       &\gamma V^j(f(x,\pi_{\phi_j}(x)))-V^j(x)\notag\\
       &\qquad\qquad+\ell(x,\pi_{\phi_j}(x)) \leq \delta_1(x),\ \forall x\in \mathcal{X},\label{cond:error1} \\
       &\gamma V^j(x^{j}_{k+1})-V^j(x^{j}_k)+\ell(x^{j}_k,u^{j}_k) \geq -\delta_2,\ \forall k\in \mathbb{N}_{\geq 0}.\label{cond:error12}
   \end{align}
\end{assumption}
We note that all the above conditions can be verified offline, and additional training can be conducted to ensure their satisfaction if they are not satisfied with the current model. The systematic way for imposing the conditions is out of the scope of this paper and will be considered in future work. 
Then, we discuss the recursive feasibility of the optimization problem (\ref{problem2}), stability of equilibrium state $x_F$, and non-increasing performance cost along with the number of iterations in the following subsections. 
\subsection{Recursive feasibility}
The recursive feasibility of problem (\ref{problem2}) can be shown based on the conventional discussion in MPC literature \cite{MPC} and properties imposed on the terminal function (\ref{cond:CLBF1})--(\ref{cond:CLBF7}). 
\begin{theo}
    Consider the system (\ref{dynamics}) controlled by the LMPC controller (\ref{MPC input}). Let Assumptions \ref{assumption:stagecost}--\ref{assumption:clbf1} hold. Then, the LMPC (\ref{problem2}) is feasible for all time instances $t\geq0$ and iterations $j\geq1$. 
\end{theo}
\begin{proof}
    We show the theorem with mathematical induction.
    First, from Assumption \ref{initialCLBF}, there exists a feasible solution of (\ref{problem2}) at $t=0$ in iteration $j$. At time $t>0$ of iteration $j$, we suppose that the problem (\ref{problem2}) is solved and the optimal solution $u_{t:t+N-1|t}^{j,*}$ and the corresponding state trajectory $x_{t+1:t+N|t}^{j,*}$ are obtained. Then, the first input $u_{t|t}^{j,*}$ is applied to the system (\ref{dynamics}), and the next state is obtained as $x_{t+1}^j=x_{t+1|t}^{j,*}$. Then, at time $t+1$ of the iteration $j$, the control input sequence $[u_{t+1}^{j,*},u_{t+2}^{j,*},\ldots, u_{t+N-1}^{j,*}, \pi_{\phi_{j-1}}(x_{t+N}^{j,*})]$ is a feasible solution of (\ref{problem2}) at time $t+1$ since $V^{j-1}(x_{t+N|t}^{j,*})\leq c$ and the set $\{x\mid V^{j-1}(x)\leq c\}$ is forward invariant under the control policy $\pi_{\phi_{j-1}}$ due to the conditions (\ref{cond:CLBF5}). This discussion holds for all $j\geq 1$.
    Thus, by mathematical induction, we can conclude that the optimization problem (\ref{problem2}) is feasible for all $t\geq 0$ and $j\geq 1$.
\end{proof}

\subsection{Stability}
Next, we discuss the stability property of the proposed method in the following theorem. 
\begin{theo}\label{theo:stability}
 Consider the system (\ref{dynamics}) controlled by the LMPC controller (\ref{MPC input}).
Let Assumptions \ref{assumption:stagecost}--\ref{assumption:adp} hold. Moreover, we assume that the error bound $\delta_1(\cdot)$ satisfies $\delta_1(x_{t+N|t}^{*,j})<\gamma^{-N} \ell(x_{t|t}^{*,j},u_{t|t}^{*,j})$, $\forall t\geq0$.
    Then, the equilibrium point $x_F$ is asymptotically stable for the closed-loop system (\ref{dynamics}) and (\ref{MPC input}).
\end{theo}
\begin{proof}
    Since the problem (\ref{problem2}) is time-invariant, we replace $J_{t\rightarrow t+N}^{\mathrm{LMPC},j}(\cdot)$ with $J_{0\rightarrow N}^{\mathrm{LMPC},j}(\cdot)$ for conciseness of the notation. Suppose we have optimal solution of (\ref{problem2}) at time $t$ as $u_{t:t+N-1|t}^{j,*}$ and $x_{t+1:t+N|t}^{j,*}$. Then, the optimal cost satisfies the following:
    \begin{align}\label{proofstability}
        &J_{t\rightarrow t+N}^{\mathrm{LMPC},j}(x_t^j)=\min_{u^j_{t|t},\ldots, u^j_{t+N-1|t}}\left[ \sum_{k=t}^{t+N-1}\gamma^k \ell(x^j_{k|t},u^j_{k|t}) \right. \notag \\
        & \qquad \qquad \qquad \left. +\gamma^{t+N} V^{j-1}(x^j_{N|t}) \right]\notag \\
        & = \gamma^t \ell(x_{t|t}^{*,j},u_{t|t}^{*,j}) + \sum_{k=t+1}^{t+N-1}\gamma^k\ell(x_{t+k|t}^{*,j},u_{t+k|t}^{*,j})\notag \\
        &\qquad \qquad +\gamma^{t+N}V^{j-1}(x_{t+N|t}^{*,j})\notag \\
        &\geq  \gamma^t\ell(x_{t|t}^{*,j},u_{t|t}^{*,j}) + \sum_{k=t+1}^{t+N-1}\gamma^k\ell(x_{t+k|t}^{*,j},u_{t+k|t}^{*,j})\notag \\
        &\qquad+\gamma^{t+N}\ell(x_{t+N|t}^{*,j},\pi_{\phi_{j-1}}(x_{t+N|t}^{*,j}))-\gamma^{t+N}\delta_1(x_{t+N|t}^{*,j}) \notag \\
        &\qquad \qquad +\gamma^{t+N+1}V^{j-1}(f(x_{t+N|t}^{*,j},\pi_{\phi_{j-1}} (x_{t+N|t}^{*,j})))\notag \\
        &\geq \gamma^t \ell(x_{t|t}^{*,j},u_{t|t}^{*,j})+J_{t\rightarrow t+N}^{\mathrm{LMPC},j}(x_{t+1|t}^{*,j})-\gamma^{t+N}\delta_1(x_{t+N|t}^{*,j}),
    \end{align}
    where we used (\ref{cond:error1}) in the first inequality.
    From (\ref{proofstability}) and the condition $\delta_1(x_{t+N|t}^{*,j})<\gamma^{-N} \ell(x_{t|t}^{*,j},u_{t|t}^{*,j})$, $\forall t\geq0$, we can conclude that the function $J_{0\rightarrow N}^{\mathrm{LMPC},j}(\cdot)$ is a decreasing Lyapunov function for the closed system (\ref{dynamics}) and (\ref{MPC input}). Thus, the final state $x_F$ is asymptotically stable.
    \begin{comment}
    \begin{align}
        J_{0\rightarrow N}^{\mathrm{LMPC},j}(x_{t+1}^j)-J_{0\rightarrow N}^{\mathrm{LMPC},j}(x_{t}^j)\leq -\ell(x_t,u_t)
    \end{align}
    \end{comment}
\end{proof}
\subsection{Performance cost}\label{subsec:performance}
Lastly, we discuss the property of the performance cost along with the number of iterations. We first define the performance cost at iteration $j$ as follows:
\begin{align}
    J_{0\rightarrow \infty}^j(x_s) = \sum_{t=0}^{\infty}\gamma^t\ell(x_t^j,u_t^j).
\end{align}
Then, we have the following theorem regarding performance cost.
\begin{theo}\label{theo:performance}
    Consider the closed loop system (\ref{dynamics}) and (\ref{MPC input}). Let Assumptions \ref{assumption:stagecost}--\ref{assumption:adp} hold. Then, the following holds:
    \begin{align}
        J_{0\rightarrow \infty}^{j-1}(x_s)\geq J_{0\rightarrow \infty}^j(x_s)-\frac{\gamma^{N}(\delta_1^{j,\mathrm{max}}+\delta_2)}{1-\gamma}.
    \end{align}
    where $\delta_1^{j,\mathrm{max}}=\max_{t}\delta_1(x_{t+N|t}^{*,j})$,
\end{theo}
\begin{proof}
    By recursively applying the condition (\ref{cond:error12}) in Assumption (\ref{assumption:adp}), the following hold for all $j\geq 1$:
    \begin{align}\label{eq:perf1}
        &V^{j-1}(x_0^{j-1})\leq \gamma V^{j-1}(x_1^{j-1}) + \ell(x_0^{j-1},u_0^{j-1})+ \delta_2\notag \\
        &\leq \gamma^N V^{j-1}(x_N^{j-1}) + \sum_{t=0}^{N-1}\gamma^t\ell(x_t^{j-1},u_t^{j-1})+ \delta_2\sum_{t=0}^{N-1}\gamma^t\notag \\
        &\leq \gamma^\infty V^{j-1}(x_\infty^{j-1}) + \sum_{t=0}^{\infty}\gamma^t\ell(x_t^{j-1},u_t^{j-1})+ \delta_2\sum_{t=0}^{\infty}\gamma^t.\notag \\
    \end{align}
    Then, from the third inequality of (\ref{eq:perf1}) and $\gamma^\infty V^{j-1}(x_\infty^{j-1})=0$, we have 
    \begin{align}\label{eq:performance3}
        &J_{0\rightarrow \infty}^{j-1}(x_s) = \sum_{t=0}^\infty \gamma^t \ell(x_t^{j-1},u_t^{j-1})\notag \\
        &\geq \sum_{t=0}^{N-1} \gamma^t \ell(x_t^{j-1},u_t^{j-1})+\gamma^N V^{j-1}(x_N^{j-1})-\delta_2\sum_{t=N}^{\infty}\gamma^t\notag \\
        &\geq \min_{u^{j-1}_0,\ldots,u^{j-1}_{N-1}}\left[ \sum_{k=0}^{N-1}\gamma^k\ell(x^{j-1}_k,u^{j-1}_k)+\gamma^NV^{j-1}(x_N^{j-1}) \right] \notag \\
        &\qquad-\frac{\delta_2\gamma^{N}}{1-\gamma}= J_{0\rightarrow N}^{\mathrm{LMPC},j}(x_0^{j})-\frac{\delta_2\gamma^{N}}{1-\gamma},
    \end{align}
    %where the first inequality follows from Assumptions \ref{assumption:adp}, \ref{assumption:clbf1}, and \ref{assumption:monotonic}.
    Moreover, from (\ref{proofstability}), we have the following
    \begin{align}\label{eq:performance1}
        &J_{0\rightarrow N}^{\mathrm{LMPC},j}(x_0^j)\geq \ell(x_0^j,u_0^j)+ J_{1\rightarrow N+1}^{\mathrm{LMPC},j}(x_1^j)-\gamma^N\delta_1^{j,\mathrm{max}}\notag \\
        & \geq \ell(x_0^j,u_0^j)+ \gamma \ell(x_1^j,u_1^j)+J_{2\rightarrow N+2}^{\mathrm{LMPC},j}(x_2^j)-\gamma^N\delta_1^{j,\mathrm{max}}\notag \\
        &\qquad-\gamma^{N+1}\delta_1^{j,\mathrm{max}} \notag\\
        &\geq \lim_{t\rightarrow \infty}\left[ \sum_{k=0}^{t-1}\gamma^k\ell(x_k^j,u_k^j)+J_{t\rightarrow t+N}^{\mathrm{LMPC},j}(x_t^j) \right]-\frac{\gamma^N\delta_1^{j,\mathrm{max}}}{1-\gamma}.
    \end{align}
    From Theorem \ref{theo:stability}, $x_F$ is asymptotically stable for the closed-loop system (\ref{dynamics}) and (\ref{MPC input}). Thus, by continuity of the function $\ell$ and Assumption \ref{assumption:stagecost}, we have 
    \begin{align}\label{eq:performance2}
        \lim_{t\rightarrow \infty} J_{t\rightarrow t+N}^{\mathrm{LMPC},j}(x_t^j) = 0.
    \end{align}
    From (\ref{eq:performance1}) and (\ref{eq:performance2}), we obtain 
    \begin{align}\label{eq:performance4}
        J_{0\rightarrow N}^{\mathrm{LMPC},j}(x_t^j)\geq  J_{0\rightarrow \infty}^j(x_s)-\frac{\gamma^N\delta_1^{j,\mathrm{max}}}{1-\gamma}.
    \end{align}
    Thus from (\ref{eq:performance3}) and (\ref{eq:performance4}), the following inequality holds:
    \begin{align}
        J_{0\rightarrow \infty}^{j-1}(x_s)&\geq J_{0\rightarrow N}^{\mathrm{LMPC},j}(x_0^j)-\frac{\delta_2\gamma^{N}}{1-\gamma}\notag \\
        &\geq J_{0\rightarrow \infty}^j(x_s)-\frac{\gamma^{N}(\delta_1^{j,\mathrm{max}}+\delta_2)}{1-\gamma}.
    \end{align}
    This proves the theorem.
\end{proof}
Theorem \ref{theo:performance} indicates that the performance cost is guaranteed to be non-increasing with respect to the number of iterations when $\delta_1^{j,\mathrm{max}}$ and $\delta_2$ (violations of the conditions (\ref{cond:CLBF5}) and (\ref{cond:CLBF6})) are zero and the upper bound of $J_{0\rightarrow \infty}^{j}(x_s)-J_{0\rightarrow \infty}^{j-1}(x_s)$ grows linearly with the maximum errors $\delta_1^{j,\mathrm{max}}$ and $\delta_2$.

\section{Simulation}\label{sec:sim}
In this section, we evaluate the proposed control scheme through numerical experiments. All experiments are conducted using Python on a Windows 11 machine with a 2.20 GHz Core i9 CPU and 32 GB of RAM. The neural networks representing the terminal components and associated control policy are implemented and trained using PyTorch. The MPC optimization problems are solved using the IPOPT solver in CasADi \cite{casadi}.
We compare the control performance and time efficiency of the proposed method with the previous reference-free iterative learning MPC method \cite{iterative1}. The method \cite{iterative1} is implemented based on the open-source repository \cite{LMPCcode}.

\begin{figure*}[tb]
 \begin{center}
  \includegraphics[width=1\hsize]{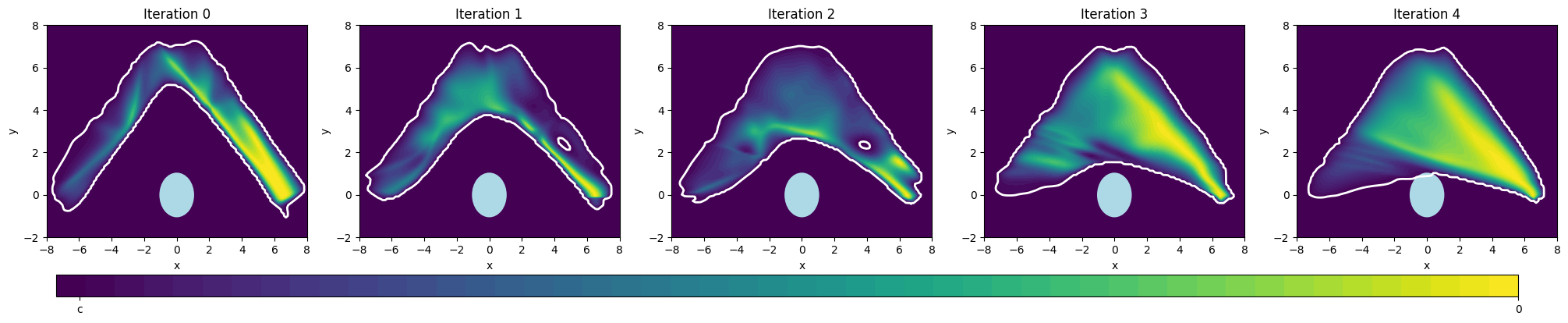}
 \end{center}
 \caption{The 2D heat map of the function $V^j$ after each iteration when the heading angle is fixed to $\theta=-\pi/4$. The white line shows the $c$ level set and light blue region shows the obstacle. }
 \label{fig:contour}
\end{figure*}

\begin{figure}[tb]
 \begin{center}
  \begin{minipage}{0.47\linewidth}
   \centering
   \includegraphics[width=\linewidth]{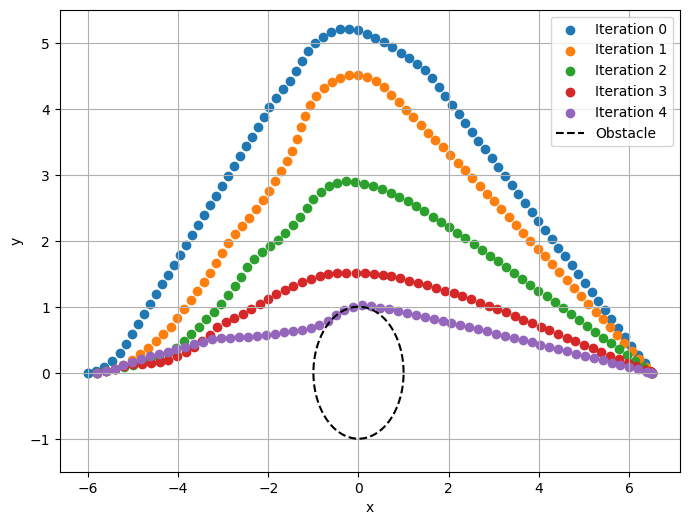}
   \caption*{(a) Proposed method}
  \end{minipage}
  \hfill
  \begin{minipage}{0.47\linewidth}
   \centering
   \includegraphics[width=\linewidth]{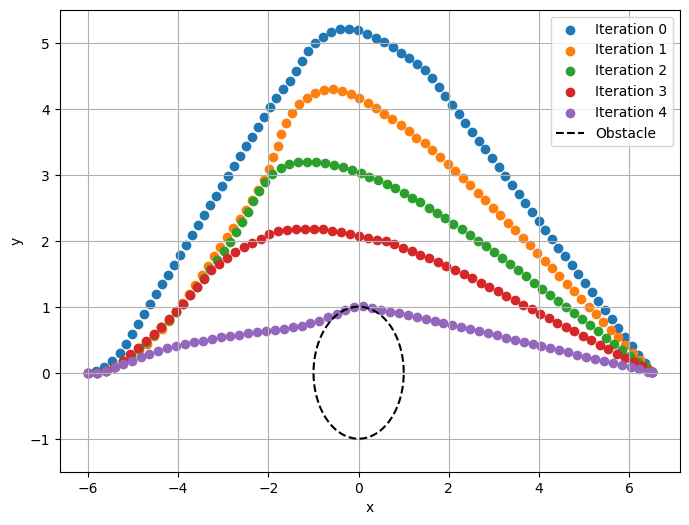}
   \caption*{(b) Previous method}
  \end{minipage}
  \caption{The actual trajectories of the vehicle position (nominal experiment).}
  \label{fig:trajectory_comparison}
 \end{center}
\end{figure}

\begin{table}[tb]
    \centering
    \caption{Total cost and computation time at each iteration (nominal experiment)}
    \begin{tabular}{llccccc}
        \toprule
        \multirow{2}{*}{\textbf{Method}} & & \multicolumn{5}{c}{\textbf{Iteration}} \\
        \cmidrule(lr){3-7}
                       &      & 0    & 1    & 2    & 3    & 4      \\
        \midrule
        \multirow{2}{*}{\textbf{Proposed}} 
                       & Cost & 5.51 & 4.09 & 3.36 & 2.99 & 2.88 \\
                       & Time [s]& --   & 2.7 & 2.5 & 2.6 & 1.8  \\
        \midrule
        \multirow{2}{*}{\textbf{Previous}} 
                       & Cost & 5.51 & 4.21 & 3.76 & 3.40 & 3.02 \\
                       & Time [s]& --   & 96 & 207 & 379 & 398 \\
        \bottomrule
    \end{tabular}
    \label{tab:cost_computation_comparison}
\end{table}

\begin{table}[tb]
    \centering
    \caption{Total cost at each iteration (TurtleBot simulation)}
    \begin{tabular}{lccccc}
        \toprule
        \multirow{2}{*}{\textbf{Method}} & \multicolumn{5}{c}{\textbf{Iteration}} \\
        \cmidrule(lr){2-6}
                       & 0    & 1    & 2    & 3    & 4       \\
        \midrule
        \textbf{Proposed}& 5.51 & 3.93 & 3.35  & 3.08  & 2.97 \\
                      
        \midrule
        \textbf{Previous}& 5.51 & 4.20 & 3.69 & 3.32  & 3.01  \\
                       
        \bottomrule
    \end{tabular}
    \label{tab:cost_computation_comparison2}
\end{table}

In this experiment, we consider the Dubins car reach-avoid problem. The vehicle dynamics are modeled as  
\begin{align}\label{eq:dubin}
    x_{k+1} = \begin{bmatrix} z_{k+1} \\ y_{k+1} \\ \theta_{k+1} \end{bmatrix} = \begin{bmatrix} z_k \\ y_k \\ \theta_k \end{bmatrix} +\Delta t \begin{bmatrix} v_k \cos \theta_k \\ v_k \sin \theta_k \\ \omega_k \end{bmatrix},
\end{align}  
where the state vector $ x_k = [z_k, y_k, \theta_k]^\top $ represents the vehicle's position $(z_k, y_k)$ and orientation $ \theta_k $. The control inputs $ u_k = [v_k, \omega_k]^\top $ consist of the velocity $ v_k $ and angular velocity $ \omega_k $. The time step is set to $ \Delta t = 0.1 $.  
The objective is to minimize the cumulative cost $ \sum_{k=1}^{\infty} 0.001\| x_k - x_F \|^2 $ while avoiding the obstacle region defined by $ (z_k - z_{\text{obs}})^2 + (y_k - y_{\text{obs}})^2 \leq 1 $ and adhering to the control constraints $ 0 \leq v_k \leq 2 $ and $ -\pi/2 \leq \omega_k \leq \pi/2 $. The goal state and initial state are set as $ x_F = [6,0,0] $ and $ x_s = [-6,0,0] $, respectively.  
The parameters used in Algorithms \ref{alg} and \ref{alg:training_CLBF_compact} are configured as $ N_{\mathrm{ite}} = 5 $, $ N = 15 $, $ \gamma = 0.8 $, $ c = 7 $, $ a_1 = a_2 = a_3 = a_4 = a_5 = 1 $, and $ k_{\mathrm{val}} = 100 $. The neural networks used to represent the function $ V^j $ and the control policy are constructed by fully connected networks with architectures 2-32-32-1 and 2-16-16-1, respectively, both utilizing the $ \tanh $ activation function. The learning rate for training the neural networks is set to $ 10^{-3}$.  
%The initial trajectory is manually constructed as the red trajectories in Fig. \ref{fig:trajectory_comparison} and the initial clbf is trained by taking samples from  as shown in Fig. \ref{}.

To evaluate the performance of the proposed method in both ideal and practical scenarios, we conduct experiments in two settings for the control problem explained above:  
(I) a nominal setting where both the controlled system and the prediction model used in the controller are given by (\ref{eq:dubin}), and  
(II) a more realistic setting where the controlled system is a TurtleBot simulated in PyBullet \cite{pybullet}, while the prediction model remains the same as (\ref{eq:dubin}). The main objectives of the experiments are threefold: (i) to verify that the proposed method can iteratively improve control performance, (ii) to demonstrate that it requires less computation during online control execution compared to the previous method \cite{iterative1}, and (iii) to show its effectiveness in a more realistic setting using the PyBullet simulation.
 For the proposed method, we learn the initial function $V^j$ by samples taken from ``behind" the initial trajectory (see the leftmost figure in Fig. \ref{fig:contour}).
%The rationale behind the second setting is to assess the robustness of the proposed method against **model mismatch**. Although PyBullet internally simulates the robot based on physical principles and detailed mechanics, the exact internal model used by the simulator is typically inaccessible or too complex to be directly used for controller design. Consequently, a simplified model (\ref{eq:dubin}) is employed for control, which inevitably introduces modeling errors. This mismatch emulates real-world scenarios where the true system dynamics are more complex than the control model.  

\begin{comment}
\textcolor{red}{For the control problem explained above}, we conduct experiments for two cases: (i) the controlled system and the prediction model used in the controllers are the same and given by (\ref{eq:dubin}), and (ii) the controlled system is a TurtleBot in PyBullet simulator \cite{pybullet} and the prediction model is (\ref{eq:dubin}). For both cases, the initial trajectory is given by blue dots in Fig. \ref{fig:trajectory_comparison}. For the proposed method, we learn the initial function $V^j$ by samples taken from ``behind" the initial trajectory (see the leftmost figure in Fig. \ref{fig:contour}).
\end{comment}

\subsection{The result of nominal experiment}
%\textcolor{red}{We first evaluate the control performance and computational efficiency during online execution for the case (I), in comparison with the previous method.} In Table \ref{tab:cost_computation_comparison}, we summarize the comparison of the results obtained from the proposed method and previous method \cite{iterative1}. We can see that the proposed method can iteratively reduce the performance cost along with the number of iterations and achieves comparable control performance to the previous method. The main advantage of the proposed method can be seen in the time required for the online computation. The proposed method is much faster than that of the comparison method, which enables real time implementations for much broader applications. %On the other hand, the proposed method requires relatively heavy offline computation .
We first evaluate the control performance and computational efficiency during online execution for the case (I), in comparison with the previous method \cite{iterative1}. Table \ref{tab:cost_computation_comparison} summarizes the results obtained by the proposed method and the previous approach. The proposed method demonstrates an iterative reduction in the performance cost with the number of iterations, achieving comparable control performance to the previous method. The main advantage of the proposed method can be seen in the time required for the online computation. The proposed method is much faster than that of the comparison method, which enables real time implementations for much broader applications.
In Fig \ref{fig:trajectory_comparison}, the state trajectories obtained by the proposed method and the comparison method \cite{iterative1} at iteration 1-5 are plotted. %Fig. \ref{} shows control sequence in each iteration and we can see that the control constraints are met by both methods.
We also show the 2D heat map of the function $V^j$ after each iteration when the heading angle is fixed to $\theta=-\pi/4$ in Fig. \ref{fig:contour}. We can see from Fig. \ref{fig:contour} that the certified region is progressively enlarged by the proposed scheme.
\subsection{The result of TurtleBot simulation}
\begin{figure}[tb]
 \begin{center}
  \begin{minipage}{0.47\linewidth}
   \centering
   \includegraphics[width=\linewidth]{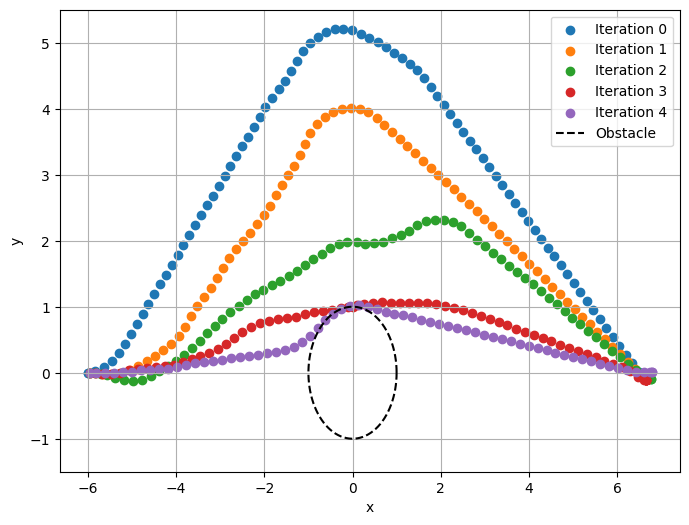}
   \caption*{(a) Proposed method}
  \end{minipage}
  \hfill
  \begin{minipage}{0.47\linewidth}
   \centering
   \includegraphics[width=\linewidth]{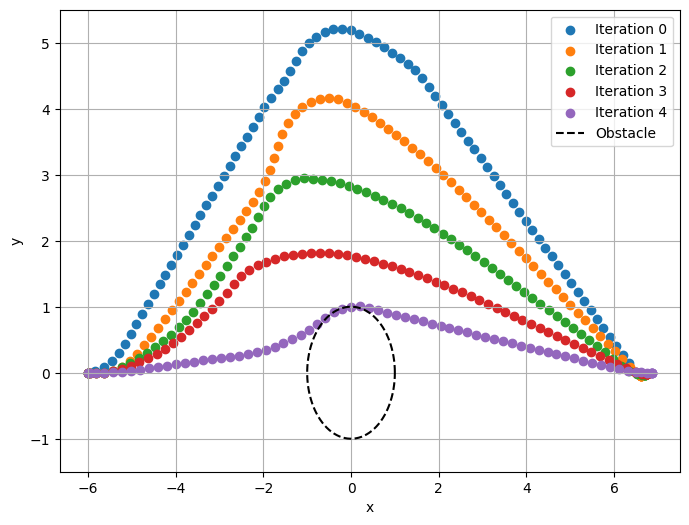}
   \caption*{(b) Previous method}
  \end{minipage}
  \caption{The actual trajectories of the vehicle position (TurtleBot simulation).}
  \label{fig:trajectory_comparison2}
 \end{center}
\end{figure}

\begin{figure}[tb]
 \begin{center}
  \includegraphics[width=1\hsize]{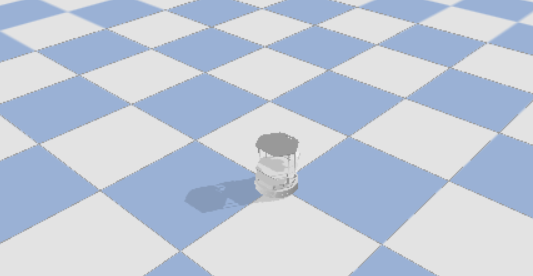}
 \end{center}
 \caption{GUI of TurtleBot in PyBullet simulator \cite{pybullet}.}
 \label{fig:bot}
\end{figure}
Next, we conduct the simulation with TurtleBot in PyBullet simulator and see whether the proposed method can be used for more realistic setting.
 %\textcolor{red}{This experiment serves to verify whether the proposed method retains its performance under modeling errors due to the discrepancy between the simplified control model (\ref{eq:dubin}) and the more detailed simulation environment.  }
The URDF of TurtleBot is accessible at \cite{turtle} and the GUI of TurtleBot in PyBullet is shown in Fig. \ref{fig:bot}.
To control the TurtleBot in PyBullet, the velocity \( v \) and angular velocity \( \omega \) of the robot must be converted into individual wheel velocities. Given the wheelbase \( L \) and wheel radius \( R \), the right and left wheel velocities, \( v_r \) and \( v_l \), are computed as follows: 
\begin{equation}
    v_r = \frac{v -\omega L / 2}{R}, \quad 
    v_l = \frac{v + \omega L / 2}{R}.
\end{equation}
For TurtleBot in PyBullet, the values of $R$ and $L$ are $R=0.035$ and $L= 0.23$.
These velocities are then applied to the TurtleBot’s wheels using PyBullet’s \texttt{setJointMotorControl2} function in velocity control mode.
Moreover, while MPC is executed with a sampling period of \( 0.1 \) s, control signals are sent at a higher frequency of \( 0.01 \) s to ensure accurate execution. The control input computed by MPC is held constant within each sampling period and applied at every simulation step.
Table \ref{tab:cost_computation_comparison2} and Fig. \ref{fig:trajectory_comparison2} show the results of the PyBullet simulation, which closely resemble those obtained in the numerical experiments.%The table confirms that the simulation yields similar results to those obtained in the numerical experiments.
\begin{comment}
\begin{equation}
\begin{aligned}
    &\min_{u_0,u_1,\ldots} \sum_{k=1}^{\infty}\| x_k-x_F \|^2\\
    &\mathrm{s.t.}\\
    &x_{k+1} = \begin{bmatrix} z_{k+1} \\ y_{k+1} \\ v_{k+1} \end{bmatrix} = \begin{bmatrix} z_k \\ y_k \\ v_k \end{bmatrix} + \begin{bmatrix} v_k \cos(\theta_k) \\ v_k \sin(\theta_k) \\ a_k \end{bmatrix}, \\
    &x_0 = [0, 0, 0]^\top, \\
    &-s \leq a_k \leq s, \\
    &\frac{(z_k - z_{\text{obs}})^2}{a_e^2} + \frac{(y_k - y_{\text{obs}})^2}{b_e^2} \geq 1, \quad \forall k \geq 0.
\end{aligned}
\end{equation}
\end{comment}
%Here, the state $x_k = [z_k, y_k, v_k]^\top$ represents the vehicle's position $(z_k, y_k)$ and velocity $v_k$, while the control inputs are acceleration $a_k$ and steering angle $\theta_k$. Constraint \((2)\) represents the acceleration saturation limit, and constraint \((3)\) ensures obstacle avoidance by maintaining a safe distance from an elliptical obstacle centered at $(z_{\text{obs}}, y_{\text{obs}})$.
%In the simulation, an initial feasible trajectory was obtained using a brute-force approach to steer the system from \(x_0\) to \(x_F\). This trajectory was used to construct the initial CLBF $V_{\theta_0}$ and the terminal cost \(Q^0(\cdot)\), which were employed to initialize the first iteration of the LMPC.

\section{Summary and Future Direction}
In this paper, we proposed a novel reference-free iterative learning MPC scheme. In the proposed method, the NN representing the terminal components of the MPC is trained with the trajectory data collected in the previous iterations to meet the conditions (\ref{cond:CLBF1})--(\ref{cond:CLBF7}).
Then, we showed that the proposed control scheme guarantees desirable properties such as recursive feasibility, stability, and non-increasing performance cost along with the number of iterations, under certain assumptions. In the simulation, we showed that the proposed scheme can iteratively improve the control performance and the time required for online control execution is largely saved compared to the previous method. In future research, we will consider extending the proposed method to uncertain or unknown dynamics settings.
\begin{comment}
\section{Declaration of generative AI and AI-assisted technologies in the writing process}
During the preparation of this work the author(s) used ChatGPT in order to improve the clarity, grammar, and conciseness of the text. After using this tool/service, the author(s) reviewed and edited the content as needed and take(s) full responsibility for the content of the publication.
\end{comment}
\bibliographystyle{IEEEtran}
\bibliography{main}

%\bibliography{IEEE} % 文献リストファイルへのパス

\vfill

\end{document}